\def\be{\begin{equation}}
\def\ee{\end{equation}}
\def\ba{\begin{array}}

\def\ea{\end{array}}

\documentclass[a4paper,11pt]{article}
\usepackage[left=3.17cm,right=3.17cm,top=2.54cm,
headheight=0.5cm,headsep=0.54cm,bottom=2.54cm,footskip=0.79cm
]{geometry}

\usepackage{amsmath,amssymb,amsthm}
\usepackage{graphicx}
\usepackage{caption}
\usepackage{subcaption}
\newtheorem{Thm}{Theorem}

\newtheorem{Cor}{Corollary}
\newtheorem{exm}{Example}

\usepackage{tikz}
\usetikzlibrary{arrows}
\usetikzlibrary{decorations.markings}
\usepackage[pdfstartview=FitH%,dvipdfm,colorlinks=true,citebordercolor={1 1 1},linkcolor=red,citecolor=blue,urlcolor=blue
]{hyperref}

\newcommand{\ket}[1]{\left|#1\right>}
\newcommand{\bra}[1]{\left<#1\right|}
\newcommand{\N}{\mathcal{N}}
\newcommand{\C}{\mathcal{C}}

\begin{document}

\title{Generalized monogamy inequalities of convex-roof extended negativity in $N$-qubit systems}

\author{Yanmin Yang$^{1}$ $\thanks{e-mail: yangym929@gmail.com}$, Wei Chen$^{2}$ $\thanks{e-mail: auwchen@scut.edu.cn}$, Gang Li$^{3}$, Zhu-Jun Zheng$^{3}$ $\thanks{e-mail: zhengzj@scut.edu.cn}$
\\[.3cm]
{\small $^{1}$School of Mathematics,  GuangZhou University, Guangzhou 510006, China}\\[.3cm]
{\small $^{2}$School of Computer Science and Network Security,}\\[.3cm]
{\small Dongguan University of Technology, Dongguan ¡¡523808, China}\\[.3cm]
{\small $^{3}$School of Mathematics, South China University of Technology,
Guangzhou 510641, China}\\[.3cm]
}
\date{}
\maketitle

\begin{abstract}
In this paper, we present some generalized monogamy inequalities based on negativity and convex-roof extended negativity (CREN). These monogamy relations are satisfied by the negativity of $N$-qubit quantum systems $ABC_1\cdots C_{N-2}$, under the partitions  $AB|C_1\cdots C_{N-2}$ and $ABC_1|C_2\cdots C_{N-2}$. Furthermore, the $W$-class states are used  to test these generalized monogamy inequalities.
\end{abstract}



\bigskip

One of the most fundamental properties of quantum correlations is that they are not shareable when distributed among many parties.
This property distinguishes the quantum correlations from classical correlations. A simple example is a pure maximally entangled state shared
between Alice and Bob. This state cannot share any additional correlation (classical or quantum) with other parties. The composite system
with a third party, say Carol, can only be a tensor product of the state of Alice and Bob with the state of Carol. This property has been
called the monogamy of entanglement and it means that the monogamy relation of entanglement is a way to characterize the different types
of entanglement distribution. The monogamy relations give rise to structures of entanglement in multipartite setting and it is important for
many tasks in quantum information theory, particularly, in quantum key distribution \cite{Paw} and quantum correlations \cite{Stre, Bai} like quantum
discord\cite{Oll}.

Although it has been shown that quantum correlation measures and entanglement measures cannot satisfy the traditional monogamy relations,
it has been shown that it does satisfy the squared concurrence $C^2$ \cite{Osb, Ye, Zhu, Luo} and the squared entanglement of formation  $E^2$  \cite{Zhu, Luo, Bai1, Bai2, Oli}. Other useful entanglement measures are negativity \cite{Vid} and convex-roof extended negativity(CREN) \cite{Lee}. The authors in \cite{Ou} showed that the monogamy inequality holds in terms of squared negativity for three-qubit states and the author in \cite{He} showed a monogamy relation conjecture on squared negativity for tripartite systems. Kim et al. showed that the squared CREN follows the monogamy inequality \cite{Kim}.

In this paper, we study the general monogamy inequalities of CREN in multi-qubit systems. We first recall some basic concepts of entanglement measures.
Then we find that the generalized monogamy inequalities always hold based on negativity and CREN in $N$-qubit systems under the partitions $AB|C_1\cdots C_{N-2}$ and $ABC_1|C_2\cdots C_{N-2}$. Detailed examples for $W$-class states are given to test the generalized monogamy inequalities.

\medskip

%We first recall the definition of concurrence, negativity and convex-roof extended negativity (CREN).
Given a bipartite pure state $|\psi\rangle_{AB}$ in a $d \otimes d^{'} (d\leq d^{'})$ quantum system, its concurrence, $C(|\psi\rangle_{AB})$ is defined as \cite{Woo}
\begin{equation}\label{Concur}
\mathcal{C}(|\psi\rangle_{AB}) = \sqrt{2[1-Tr(\rho_{A}^{2})]} = \sqrt{2[1-Tr(\rho_{B}^{2})]},
\end{equation}
where $\rho_{A}$ is reduced density matrix by tracing over the subsystem $B,$ $\rho_{A} = Tr_{B}(|\psi\rangle_{AB}\langle\psi|)$ (and analogously for $\rho_{B}$).

For any mixed state $\rho_{AB},$ the concurrence is given by the minimum average concurrence taken over all decompositions $\{p_i,\ket{\psi_i}\}$ of $\rho_{AB},$
the so-called convex roof \cite{Uhl}
\begin{equation}
\mathcal{C}(\rho_{AB}) = \min_{\{p_{i},|\psi_{i}\rangle\}}\sum_{i}p_{i}\mathcal{C}(|\psi_{i}\rangle).
\end{equation}
%where the convex roof is notoriously hard to evaluate and therefore it is difficult to determine whether or not an arbitrary state is entangled.

The concurrence of assistance (COA)  of any mixed state $\rho_{AB}$ is defined as \cite{Lau}
\begin{eqnarray}
\mathcal{C}_{a}(\rho_{AB})  &=& \max_{\{p_{i},|\psi_{i}\rangle\}}\sum_{i}p_{i}\mathcal{C}(|\psi_{i}\rangle),
\end{eqnarray}
where the maximum is taken over all possible pure state decompositions $\{p_{i},|\psi_{i}\rangle\}$ of $\rho_{AB}$.

If $\rho_{AB}$ be a two-qubit state, then the COA is defined by \cite{Lau}, \cite{YS}
\begin{eqnarray}
  \mathcal{C}_{a}(\rho_{AB})  &=& Tr(\sqrt{\rho_{AB}\widetilde{\rho}_{AB}})\\ \label{Ca}
  &=& \sum_i \lambda i,
\end{eqnarray}
where $\widetilde{\rho}_{AB}=(\sigma_y\otimes \sigma_y)\rho_{AB}^*(\sigma_y\otimes \sigma_y)$, $\sigma_y$ is Pauli matrix and $\rho_{AB}^*$  is complex conjugation of $\rho_{AB}$ taken in the standard basis, and
\begin{equation}\label{C}
   \mathcal{C}(\rho_{AB}) = \max \{0, \lambda_1-\sum_{i>1}\lambda_i \},
\end{equation}
is the concurrence of $\rho_{AB}$ with $\lambda_i$
being the square roots of the eigenvalues of $\rho_{AB}\widetilde{\rho}_{AB}$ in decreasing
order.

Another well-known quantification of bipartite entanglement is negativity \cite{Vid}, which is based on the positive partial transposition (PPT)
criterion \cite{Per, Hor}. For a bipartite state $\rho_{AB}$ in a $d \otimes d^{'} (d\leq d^{'})$ quantum system, its negativity is defined as
\begin{equation}
\mathcal{N}(\rho_{AB}) = \|\rho_{AB}^{T_A}\|-1,
\end{equation}
where $\rho_{AB}^{T_A}$ is the partial transpose with respect to the subsystem $A$ and $\|X\|$ denotes the trace norm of $X,$ i.e. $\|X\| = Tr \sqrt{XX^{\dagger}}.$

In a $d \otimes d^{'} (d\leq d^{'})$ quantum system, if a bipartite pure state $|\psi\rangle_{AB}$ with the Schmidt decomposition,
\begin{equation}
|\psi\rangle_{AB} = \sum_{i=0}^{d-1}\sqrt{\lambda_i}|ii\rangle,~~~ \lambda_i\geq 0, ~~~~\sum_{i=0}^{d-1}\lambda_i = 1,
\end{equation}
then \cite{Kim}
\begin{equation}\label{N}
\mathcal{N}(\rho_{AB}) = 2\sum_{i< j}\sqrt{\lambda_i\lambda_j}.
\end{equation}

To overcome the lack of separability criterion, one modification of negativity is convex-roof extended negativity (CREN), which gives a perfect discrimination of PPT bound entangled states
and separable states in any bipartite quantum system. For a bipartite mixed state $\rho_{AB},$ CREN is defined as
\begin{equation}
\widetilde{\mathcal{N}}(\rho_{AB}) = \min_{\{p_{i},|\psi_{i}\rangle\}}\sum_{i}p_{i}\mathcal{N}(|\psi_{i}\rangle),
\end{equation}
where the minimum is taken over all possible pure state decompositions $\{p_{i},|\psi_{i}\rangle\}$ of $\rho_{AB}$.

Similar to the duality between concurrence and COA, we can also define a dual of CREN, namely convex-roof extended negativity of assistance  (CRENOA), by taking the maximum value of average negativity over all possible pure state decomposition $\{p_{i},|\psi_{i}\rangle\}$ of mixed state $\rho_{AB}$, i.e.
\begin{equation}
\widetilde{\mathcal{N}}_{a}(\rho_{AB}) = \max_{\{p_{i},|\psi_{i}\rangle\}}\sum_{i}p_{i}\mathcal{N}(|\psi_{i}\rangle).
\end{equation}

CREN is equivalent to concurrence for any pure state with Schmidt rank two \cite{Kim}. It follows that
for any two-qubit mixed state $\rho_{AB}$,
\begin{eqnarray}\label{c=n}
\C(\rho_{AB}) =\min_{\{p_{i},|\psi_{i}\rangle\}}\sum_{i}p_{i}\mathcal{C}(|\psi_{i}\rangle)= \min_{\{p_{i},|\psi_{i}\rangle\}}\sum_{i}p_{i}\widetilde{\mathcal{N}}(|\psi_{i}\rangle)= \widetilde{\N}(\rho_{AB}),\\ \label{ca=na}
\C_a(\rho_{AB})=\max_{\{p_{i},|\psi_{i}\rangle\}}\sum_{i}p_{i}\mathcal{C}(|\psi_{i}\rangle)=
\max_{\{p_{i},|\psi_{i}\rangle\}}\sum_{i}p_{i}\widetilde{\mathcal{N}}(|\psi_{i}\rangle)=\widetilde{\N}_a(\rho_{AB}).
\end{eqnarray}

For any $N$-qubit pure state $\ket{\psi}_{A|B_1\cdots B_{N-1}}$, it has been
shown that the concurrence and COA of $\ket{\psi}_{A|B_1\cdots B_{N-1}}$ satisfy
monogamy inequalities \cite{Osb,GBS}:
\begin{equation}
 \sum_{i=1}^{N-1}C^2(\rho_{AB_i}) \leq  C^2(\ket{\psi}_{A|B_1\cdots B_{N-1}}) \leq  \sum_{i=1}^{N-1}C_a^2(\rho_{AB_i}),
\end{equation}
where $\rho_{AB_i}=Tr_{B_1\cdots B_{i-1}B_{i+1}\cdots B_{N-1}}(\ket{\psi}_{A|B_1\cdots B_{N-1}}\bra \psi)$.

Combining   with Eq.(\ref{c=n}) and Eq.(\ref{ca=na}), we have
\begin{equation}\label{NCNa}
 \sum_{i=1}^{N-1}\widetilde{\N}^2(\rho_{AB_i})
  \leq   C^2(\ket{\psi}_{A|B_1\cdots B_{N-1}}) \leq  \sum_{i=1}^{N-1}\widetilde{\N}_a^2(\rho_{AB_i}).
\end{equation}

The concurrence is related to the linear entropy of a state \cite{SF},
\begin{equation}\label{}
 T(\rho)=1-Tr(\rho^2).
\end{equation}
Given a bipartite state $\rho$, $T(\rho)$ has the property \cite{ZGZG},
\begin{equation}\label{T}
   T(\rho_A)+ T(\rho_B)\geq  T(\rho_{AB})\geq  |T(\rho_A)- T(\rho_B)|.
\end{equation}

From the definition of pure state concurrence together with Eq.(\ref{T}), we have
\begin{eqnarray}\label{Cg}
 C^2(\ket{\psi}_{A|BC_1 \cdots C_{N-2}})+C^2(\ket{\psi}_{B|AC_1 \cdots C_{N-2}}) \geq  C^2(\ket{\psi}_{AB|C_1 \cdots C_{N-2}}), \\ \label{Cl}
| C^2(\ket{\psi}_{A|BC_1 \cdots C_{N-2}})-C^2(\ket{\psi}_{B|AC_1 \cdots C_{N-2}})| \leq C^2(\ket{\psi}_{AB|C_1 \cdots C_{N-2}}).
\end{eqnarray}

For an $N$-qubit  pure state $\ket{\psi}_{ABC_1\cdots C_{N-2}}$, the  negativity $\mathcal{N}(|\psi\rangle_{AB|C_1\cdots C_{N-2}})$ of the state $|\psi\rangle_{AB|C_1\cdots C_{N-2}}$, viewed as a bipartite state with partition $AB|C_1\cdots C_{N-2}$, satisfies the following monogamy inequalities.

\begin{Thm}
For any $N$-qubit  pure state $|\psi\rangle_{ABC_1\cdots C_{N-2}},$ we have
\begin{equation}\label{Thm}
\mathcal{N}^2(|\psi\rangle_{AB|C_1\cdots C_{N-2}}) \geq \max \{\sum_{i=1}^{N-2}[\widetilde{\mathcal{N}}^2(\rho_{AC_i})-\widetilde{\mathcal{N}}_{a}^2(\rho_{BC_i})], \sum_{i=1}^{N-2}[\widetilde{\mathcal{N}}^2(\rho_{BC_i})-\widetilde{\mathcal{N}}^2_{a}(\rho_{AC_i})]\},
\end{equation}
where $\rho_{AB} = Tr_{C_1\cdots C_{N-2}}(|\psi\rangle\langle\psi|),$  $\rho_{AC_i} = Tr_{BC_1 \cdots C_{i-1}C_{i+1}\cdots C_{N-2}}(|\psi\rangle\langle\psi|)$ and $\rho_{BC_i} = Tr_{AC_1 \cdots C_{i-1}C_{i+1}\cdots C_{N-2}}(|\psi\rangle\langle\psi|).$
\end{Thm}

\begin{proof}

Let $|\psi\rangle_{ABC_1\cdots C_{N-2}}$ be a $N$-qubit  pure state, then we have a Schmidt decomposition
$|\psi\rangle_{AB|C_1\cdots C_{N-2}} = \sum_{i=0}^{3}\sqrt{\lambda_i}|ii\rangle.$ Then from Eq.(\ref{Concur}), we get
\begin{equation}
\mathcal{C}(|\psi\rangle_{AB|C_1\cdots C_{N-2}}) = \sqrt{2(1-Tr\rho_{AB}^{2})}
\end{equation}
where
\begin{equation*}
  \rho_{AB} = Tr_{C_1\cdots C_{N-2}}(|\psi\rangle_{AB|C_1\cdots C_{N-2}}\langle\psi|)
   = \sum_{i=0}^{3}\lambda_i|i\rangle\langle i|.
\end{equation*}
We thus obtain
\begin{equation}
 \mathcal{C}(|\psi\rangle_{AB|C_1C_2\cdots C_{N-2}}) =2 \sqrt{\sum_{i<j}\lambda_i\lambda_j}.
\end{equation}
From Eq.(\ref{N}), we have
\begin{equation}
\mathcal{N}(|\psi\rangle_{AB|C_1C_2\cdots C_{N-2}}) = 2\sum_{i<j}\sqrt{\lambda_i\lambda_j}.
\end{equation}
Consequently, we have
\begin{eqnarray*}
  \mathcal{N}^2(|\psi\rangle_{AB|C_1\cdots C_{N-2}}) & \geq & \mathcal{C}^2(|\psi\rangle_{AB|C_1\cdots C_{N-2}}) \\
  &\geq&  | C^2(\ket{\psi}_{A|BC_1 \cdots C_{N-2}})-C^2(\ket{\psi}_{B|AC_1 \cdots C_{N-2}})| \\
&\geq&\max\{\sum_{i=1}^{N-2}[\widetilde{\mathcal{N}}^2(\rho_{AC_i})-\widetilde{\mathcal{N}}_{a}^2(\rho_{BC_i})], \sum_{i=1}^{N-2}[\widetilde{\mathcal{N}}^2(\rho_{BC_i})-\widetilde{\mathcal{N}}^2_{a}(\rho_{AC_i})]\},
\end{eqnarray*}
where the second inequality is due to   Eq.(\ref{Cl}), the third inequality is due to Eq.(\ref{NCNa}).
\end{proof}

%Theorem 1 shows that the entanglement contained in the pure states $|\psi\rangle_{ABC_1C_2\cdots C_{N-2}}$ is related to the sum of entanglements between bipartitions of the system.
A monogamy-type lower bound of $\mathcal{N}(|\psi\rangle_{AB|C_1\cdots C_{N-2}})$ is given by Theorem 1.  According to the relation between negativity and concurrence, we will give an upper bound of $\mathcal{N}(|\psi\rangle_{AB|C_1\cdots C_{N-2}})$.

\begin{Thm}
For any $N$-qubit pure state $|\psi\rangle_{ABC_1\cdots C_{N-2}},$ we have
\begin{equation}\label{Thm2}
\mathcal{N}^2(|\psi\rangle_{AB|C_1\cdots C_{N-2}}) \leq \frac{r(r-1)}{2}[2\widetilde{\mathcal{N}}_a^2(\rho_{AB})+\sum_{i=1}^{N-2}(\widetilde{\mathcal{N}}_a^2(\rho_{AC_i})+\widetilde{\mathcal{N}}_{a}^2(\rho_{BC_i})],
\end{equation}
where  $r$ is the Schmidt rank  of the pure state $|\psi\rangle_{AB|C_1\cdots C_{N-2}}$, $\rho_{AB} = Tr_{C_1\cdots C_{N-2}}(|\psi\rangle\langle\psi|)$,  $\rho_{AC_i} = Tr_{BC_1 \cdots C_{i-1}C_{i+1}\cdots C_{N-2}}(|\psi\rangle\langle\psi|)$ and $\rho_{BC_i} = Tr_{AC_1 \cdots C_{i-1}C_{i+1}\cdots C_{N-2}}(|\psi\rangle\langle\psi|)$.
\end{Thm}

\begin{proof}
From Eq. (32) in \cite{Elt},
we have
\begin{equation}\label{1}
\mathcal{N}(|\psi\rangle_{AB|C_1\cdots C_{N-2}}) \leq \sqrt{\frac{r(r-1)}{2}}\mathcal{C}(|\psi\rangle_{AB|C_1\cdots C_{N-2}}).
\end{equation}
In addition, we have the fact that
\begin{eqnarray}
%\nonumber to remove numbering (before each equation)
  \C^2(|\psi\rangle_{AB|C_1\cdots C_{N-2}})  &\leq & C^2(\ket{\psi}_{A|BC_1\cdots C_{N-2}})+C^2(\ket{\psi}_{B|AC_1\cdots C_{N-2}}) \\ \label{2}
   &\leq&  2\widetilde{\N}_a^2(\rho_{AB})+\sum_{i=1}^{N-2}(\widetilde{\N}_a^2(\rho_{AC_i})+\widetilde{\N}_{a}^2(\rho_{BC_i})),
\end{eqnarray}
where the first inequality is due to   Eq.(\ref{Cg}), the second inequality is due to the right inequality of Eq.(\ref{NCNa}).

From inequalities Eq.(\ref{1}) and Eq.(\ref{2}), the inequality Eq.(\ref{Thm2}) can be deserved.
\end{proof}

\begin{Cor}
If the Schmidt rank of  pure state $|\psi\rangle_{AB|C_1\cdots C_{N-2}}$ is two, then we have
\begin{equation}\label{Cor1}
\mathcal{N}^2(|\psi\rangle_{AB|C_1\cdots C_{N-2}}) \leq 2\widetilde{\mathcal{N}}_a^2(\rho_{AB})+\sum_{i=1}^{N-2}(\widetilde{\mathcal{N}}_a^2(\rho_{AC_i})+\widetilde{\mathcal{N}}_{a}^2(\rho_{BC_i}),
\end{equation}
where $\rho_{AB} = Tr_{C_1\cdots C_{N-2}}(|\psi\rangle\langle\psi|),$  $\rho_{AC_i} = Tr_{BC_1 \cdots C_{i-1}C_{i+1}\cdots C_{N-2}}(|\psi\rangle\langle\psi|),$ and $\rho_{BC_i} = Tr_{AC_1 \cdots C_{i-1}C_{i+1}\cdots C_{N-2}}(|\psi\rangle\langle\psi|).$
\end{Cor}

\begin{exm}
Consider   the $N$-qubit generalized $W$-class states \cite{Kim1}:
\begin{equation}\label{w-state}
|W\rangle_{A_1A_2\cdots A_N} = a_1|10\cdots 0\rangle_{A_1A_2\cdots A_N}+a_2|01\cdots 0\rangle_{A_1A_2\cdots A_N}+\cdots+a_N|00\cdots 1\rangle_{A_1A_2\cdots A_N},
\end{equation}
where $\sum_{i=1}^{N}|a_i|^2=1.$
The state $|W\rangle_{A_1A_2|A_3 \cdots A_N}$, viewed as a bipartite state, has the form
\begin{eqnarray*}
  |W\rangle_{A_1A_2|A_3 \cdots A_N}=\sqrt{|a_1|^2+|a_2|^2}(\frac{a_1}{\sqrt{|a_1|^2+|a_2|^2}}\ket{10}+\frac{a_2}
  {\sqrt{|a_1|^2+|a_2|^2}}\ket{01})\otimes \ket{0\cdots 0}+\\
  (\sum\limits_{i=3}^N |a_i|^2 )^{\frac{1}{2}}\ket{00}\otimes (\frac{a_3}{(\sum\limits_{i=3}^N |a_i|^2 )^{\frac{1}{2}}}\ket{10\cdots 0}+\cdots+\frac{a_N}{(\sum\limits_{i=3}^N |a_i|^2 )^{\frac{1}{2}}}\ket{00\cdots 1}).
\end{eqnarray*}
Hence, $\N^2(|W\rangle_{A_1A_2|A_3 \cdots A_N})=4(|a_1|^2+|a_2|^2)\sum\limits_{i=3}^N |a_i|^2 $.

For any $1\leq i\neq j\leq N$, we have
\begin{equation*}
  \rho_{A_iA_j}=(a_i\ket{10}+a_j\ket{01})(a_i^*\bra{10}+a_j^*\bra{01})+\sum_{k\neq i,j}|a_k|^2\ket{00}\bra{00}.
\end{equation*}
Furthermore, from Eqs. (\ref{Ca}), (\ref{C}), (\ref{c=n}) and (\ref{ca=na}), we have
\begin{equation*}
\widetilde{\N}(\rho_{A_iA_j})=\widetilde{\N}_a(\rho_{A_iA_j})=2|a_i||a_j|.
\end{equation*}
The lower bound of $\N^2(|W\rangle_{A_1A_2|A_3 \cdots A_N})$, that is, the right hand side of Eq.(\ref{Thm}) is equal to $4||a_1|^2-|a_2|^2 |\sum\limits_{i=3}^N |a_i|^2$. And the upper bound in Eq.(\ref{Thm2}) is equal to $8|a_1|^2|a_2|^2+4(|a_1|^2+|a_2|^2)\sum\limits_{i=3}^N |a_i|^2 $.

When either $a_1=0$ or $a_2=0$, the lower bound of $\N^2(|W\rangle_{A_1A_2|A_3 \cdots A_N})$ is equal to upper bound.

For $N=3$, $a_3=\frac{1}{\sqrt{3}}$, suppose $|a_1|^2>|a_2|^2$, then the lower and upper bounds of $\N^2(|W\rangle_{A_1A_2|A_3})$ are shown in the following figure:
\begin{figure}
  \centering
   \begin{subfigure}[b]{0.4\textwidth}
                \centering
                \includegraphics[width=\textwidth]{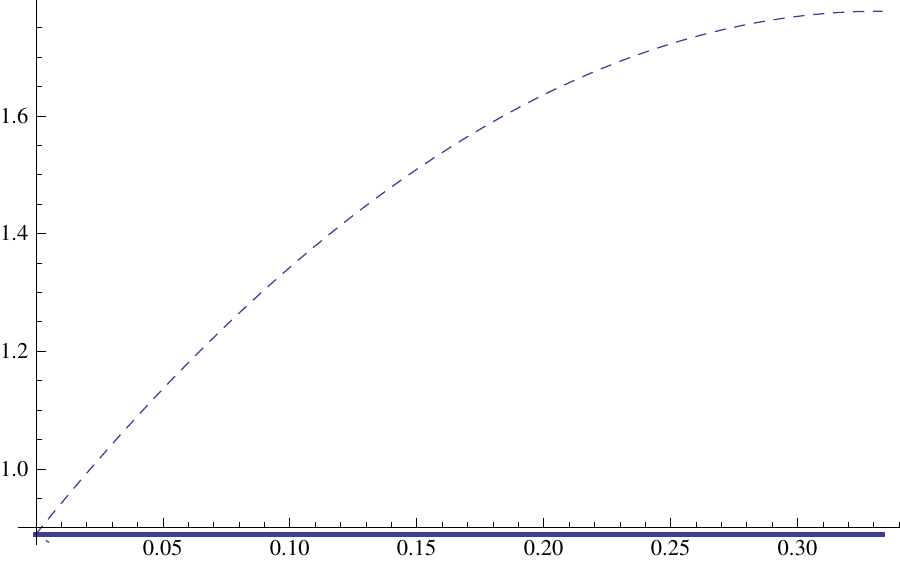}
                \caption{The dotted line is the upper bound of $\N^2(|W\rangle_{A_1A_2|A_3})$, the solid line is $\N^2(|W\rangle_{A_1A_2|A_3})$. And the abscissa represents the value range of $|a_2|^2$ from $0$ to $\frac{1}{3}$.}
                \label{fig1}
        \end{subfigure}%
         \qquad \qquad
  \begin{subfigure}[b]{0.4\textwidth}
                \centering
                \includegraphics[width=\textwidth]{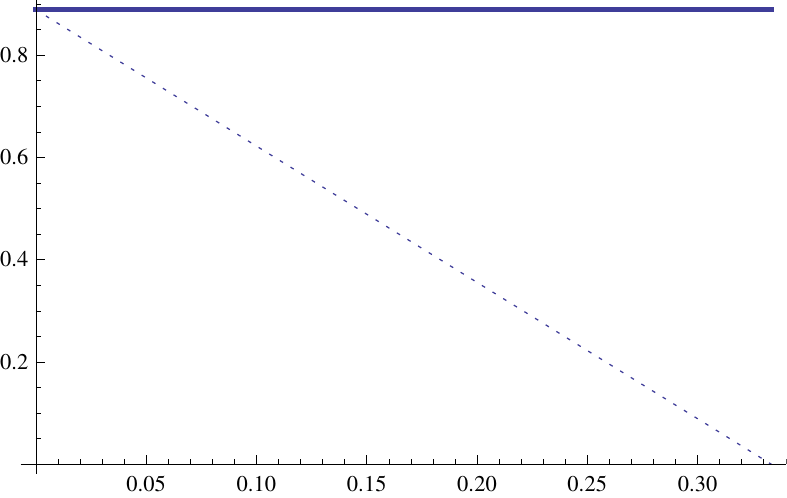}
                \caption{The dotted line is the lower bound of $\N^2(|W\rangle_{A_1A_2|A_3})$, the solid line is $\N^2(|W\rangle_{A_1A_2|A_3})$. And the abscissa represents the value range of $|a_2|^2$ from $0$ to $\frac{1}{3}$.}
                \label{fig2}
        \end{subfigure}%
  \caption{The monogamy relation of $\N^2(|W\rangle_{A_1A_2|A_3})$.}\label{}
\end{figure}

\end{exm}

\begin{Cor}
For any $N$-qubit pure state $\ket{\psi}_{ABC_1\cdots C_{N-2}}$, if the Schmidt rank of state $|\psi\rangle_{AB|C_1\cdots C_{N-2}}$ is two, then
\begin{itemize}
  \item[(i).]  we have monogamy relations
\begin{eqnarray}
  \nonumber
     |\mathcal{N}^2(|\psi\rangle_{A|BC_1\cdots C_{N-2}}) - \mathcal{N}^2(|\psi\rangle_{B|AC_1\cdots C_{N-2}})| & \leq &  \mathcal{N}^2(|\psi\rangle_{AB|C_1\cdots C_{N-2}}) \\ \label{cor2i}
     \leq  \mathcal{N}^2(|\psi\rangle_{A|BC_1\cdots C_{N-2}}) + \mathcal{N}^2(|\psi\rangle_{B|AC_1\cdots C_{N-2}}). &&
  \end{eqnarray}
Specially, if the  systems $B$ and $AC_1\cdots C_{N-2}$ are not entangled, both the two equalities hold.
  \item[(ii).]  the three terms $\mathcal{N}^2(|\psi\rangle_{A|BC_1\cdots C_{N-2}}),$ $\mathcal{N}^2(|\psi\rangle_{B|AC_1\cdots C_{N-2}})$ and   $\mathcal{N}^2(|\psi\rangle_{AB|C_1\cdots C_{N-2}})$ have following relations:
\begin{eqnarray}\label{r1}
  \mathcal{N}^2(|\psi\rangle_{A|BC_1\cdots C_{N-2}})&\leq  & \mathcal{N}^2(|\psi\rangle_{B|AC_1\cdots C_{N-2}})+\mathcal{N}^2(|\psi\rangle_{AB|C_1\cdots C_{N-2}}), \\ \label{r2}
  \mathcal{N}^2(|\psi\rangle_{B|AC_1\cdots C_{N-2}})&\leq & \mathcal{N}^2(|\psi\rangle_{A|BC_1\cdots C_{N-2}})+\mathcal{N}^2(|\psi\rangle_{AB|C_1\cdots C_{N-2}}), \\ \label{r3}
  \mathcal{N}^2(|\psi\rangle_{AB|C_1\cdots C_{N-2}}) &\leq&\mathcal{N}^2(|\psi\rangle_{A|BC_1\cdots C_{N-2}})+ \mathcal{N}^2(|\psi\rangle_{B|AC_1\cdots C_{N-2}}).
\end{eqnarray}
\end{itemize}
\end{Cor}

For the $W$-class states (\ref{w-state}), we have $\mathcal{N}^2(|\psi\rangle_{A_i|A_1\cdots \widehat{A}_i \cdots A_N}) = 4a_i^2(\sum_{j\neq i}^{N}a_j^2)$ for any $1\leq i \leq N$. Clearly, relations (\ref{r1})-(\ref{r3}) are satisfied. And the first inequality (second inequality) of Eq.(\ref{cor2i}) is just the Eq.(\ref{Thm}) in Theorem 1 (Eq.(\ref{Thm2}) in Theorem 2).

The above results can generalized to the negativity $\mathcal{N}(|\psi\rangle_{ABC_1|C_2\cdots C_{N-2}})$ under partition $ABC_1|C_2\cdots C_{N-2}$
for pure state $|\psi\rangle_{ABC_1C_2\cdots C_{N-2}}$.

\begin{Thm}
For any $N$-qubit pure state $|\psi\rangle_{ABC_1C_2\cdots C_{N-2}},$ we have
\begin{equation}\label{thm3 1}
\begin{array}{lll}
 \mathcal{N}^2(|\psi\rangle_{ABC_1|C_2\cdots C_{N-2}})  &\geq & \frac{2}{r(r-1)}\max\{\sum\limits_{i=1}^{N-2}[\widetilde{\mathcal{N}}^2(\rho_{AC_i})-
 \widetilde{\mathcal{N}}_{a}^2(\rho_{BC_i})],\\
 &&\sum\limits_{i=1}^{N-2}[\widetilde{\N}^2(\rho_{BC_i})-
 \widetilde{\N}^2_{a}(\rho_{AC_i})]\}-\sum\limits_{j\in J}\widetilde{\N}^2_{a}(\rho_{C_1j}),
\end{array}
\end{equation}

\begin{equation}\label{thm3 2}
\begin{array}{lll}
 \mathcal{N}^2(|\psi\rangle_{ABC_1|C_2\cdots C_{N-2}})  &\geq&
 \sum\limits_{j\in J}\widetilde{\N}^2(\rho_{C_1j})-\frac{r(r-1)}{2}[2\widetilde{\mathcal{N}}_a^2(\rho_{AB})\\
 &&+\sum\limits_{i=1}^{N-2}(\widetilde{\mathcal{N}}_a^2(\rho_{AC_i})+\widetilde{\mathcal{N}}_{a}^2(\rho_{BC_i})],
\end{array}
\end{equation}
where $J=\{A, B, C_2,\cdots, C_{N-2}\}$, $r$ is the Schmidt rank  of the pure state $|\psi\rangle_{AB|C_1\cdots C_{N-2}}$, and $\rho_{C_1j}$ is the reduced density matrix by tracing over the subsystems except for $C_1$ and $j$.
\end{Thm}

\begin{proof}
 For any $N$-qubit pure state $|\psi\rangle_{ABC_1C_2\cdots C_{N-2}},$ we have a Schmidt decomposition
$|\psi\rangle_{ABC_1|C_2\cdots C_{N-2}} = \sum_{i=0}^{6}\sqrt{\lambda_i}|ii\rangle.$ Then from Eq.(\ref{Concur}), we get
\begin{equation}
\mathcal{C}(|\psi\rangle_{ABC_1|C_2\cdots C_{N-2}}) = \sqrt{2(1-Tr\rho_{ABC_1}^{2})}
\end{equation}
where
$\rho_{ABC_1} = Tr_{C_2\cdots C_{N-2}}(|\psi\rangle_{ABC_1|C_2\cdots C_{N-2}}\langle\psi|)
   = \sum_{i=0}^{6}\lambda_i|i\rangle\langle i|$.

Hence we obtain
\begin{equation}
\mathcal{C}(|\psi\rangle_{ABC_1|C_2\cdots C_{N-2}}) =2 \sqrt{\sum_{i<j}\lambda_i\lambda_j}.
\end{equation}
From Eq.(\ref{N}),
\begin{equation}
\mathcal{N}(|\psi\rangle_{ABC_1|C_2\cdots C_{N-2}}) = 2(\sum_{i<j}\sqrt{\lambda_i\lambda_j}),
\end{equation}
so we get
\begin{eqnarray*}
  \mathcal{N}^2(|\psi\rangle_{ABC_1|C_2\cdots C_{N-2}}) &\geq & \mathcal{C}^2(|\psi\rangle_{ABC_1|C_2\cdots C_{N-2}}) \\
   &=& 2(1-Tr(\rho_{ABC_1}^2))\\
   &\geq& |C^2(\ket{\psi}_{AB|C_1\cdots C_{N-2}})-C^2(\ket{\psi}_{C_1|ABC_2\cdots C_{N-2}})|.
\end{eqnarray*}
If $C^2(\ket{\psi}_{AB|C_1\cdots C_{N-2}})>C^2(\ket{\psi}_{C_1|ABC_2\cdots C_{N-2}})$, then
\begin{eqnarray*}
\mathcal{N}^2(|\psi\rangle_{ABC_1|C_2\cdots C_{N-2}}) &\geq &C^2(\ket{\psi}_{AB|C_1\cdots C_{N-2}})-C^2(\ket{\psi}_{C_1|ABC_2\cdots C_{N-2}})\\
&\geq& \frac{2}{r(r-1)}\widetilde{\N}^2(\ket{\psi}_{AB|C_1\cdots C_{N-2}})- \sum\limits_{j\in J}\widetilde{\N}_a^2(\rho_{C_1j}),
\end{eqnarray*}
where the second inequality is due to   Eq.(\ref{1}) and Eq.(\ref{NCNa}).
Combine with  Eq.(\ref{Thm}), we can obtain the inequality (\ref{thm3 1}).

If $C^2(\ket{\psi}_{AB|C_1\cdots C_{N-2}})<C^2(\ket{\psi}_{C_1|ABC_2\cdots C_{N-2}})$, then
\begin{eqnarray*}
\mathcal{N}^2(|\psi\rangle_{ABC_1|C_2\cdots C_{N-2}}) &\geq & C^2(\ket{\psi}_{C_1|ABC_2\cdots C_{N-2}})-C^2(\ket{\psi}_{AB|C_1\cdots C_{N-2}})\\
&\geq& \sum\limits_{j\in J}\widetilde{\N}^2(\rho_{C_1j})-\widetilde{\N}^2(\ket{\psi}_{AB|C_1\cdots C_{N-2}}),
\end{eqnarray*}
where the second inequality is due to   Eq.(\ref{NCNa}).
Combine with  Eq.(\ref{Thm2}),  the inequality (\ref{thm3 2}) holds.
\end{proof}

Similar to Theorem 2, we also have an upper bound of $\mathcal{N}^2(|\psi\rangle_{ABC_1|C_2\cdots C_{N-2}})$.
\begin{Thm}
For any $N$-qubit pure state $|\psi\rangle_{ABC_1C_2\cdots C_{N-2}},$ we have
\begin{eqnarray}\label{thm4}\nonumber
\mathcal{N}^2(|\psi\rangle_{ABC_1|C_2\cdots C_{N-2}}) &\leq & \frac{r(r-1)}{2}[2\widetilde{\mathcal{N}}_a^2(\rho_{AB})+
\sum_{i=1}^{N-2}\widetilde{\mathcal{N}}_a^2(\rho_{AC_i})\\
&&+\sum_{i=1}^{N-2}\widetilde{\mathcal{N}}_{a}^2(\rho_{BC_i})+\sum_{j\in J}\widetilde{\mathcal{N}}^2_{a}(\rho_{C_1j})],
\end{eqnarray}
where $J$ and $\rho_{C_1j}$ are defined as in Theorem 3, and $r$ is the Schmidt rank of the pure state $|\psi\rangle_{ABC_1|C_2\cdots C_{N-2}}.$
\end{Thm}
\begin{proof}
From Eq.(32) in \cite{Elt},
we have
\begin{eqnarray*}\label{}
\mathcal{N}^2(|\psi\rangle_{ABC_1|C_2\cdots C_{N-2}}) &\leq & \frac{r(r-1)}{2}\mathcal{C}^2(|\psi\rangle_{ABC_1|C_2\cdots C_{N-2}})\\
&\leq& \frac{r(r-1)}{2}( C^2(\ket{\psi}_{AB|C_1\cdots C_{N-2}})+C^2(\ket{\psi}_{C_1|ABC_2\cdots C_{N-2}})).
\end{eqnarray*}
Combine with Eq.(\ref{Thm2}) and Eq.(\ref{NCNa}), the inequality (\ref{thm4}) can be deserved.
\end{proof}

\begin{exm}
For the  $N$-qubit generalized W-class states (\ref{w-state}), we have
\begin{equation}\label{}
  \N^2(|W\rangle_{A_1A_2A_3|A_4 \cdots A_N})=4\sum\limits_{i=1}^3 |a_i|^2\sum\limits_{j=4}^N |a_j|^2.
\end{equation}
The lower bound of $\N^2(|W\rangle_{A_1A_2A_3|A_4 \cdots A_N})$ (Eq.(\ref{thm3 2})) is
$$4||a_3|^2-|a_1|^2-|a_2|^2|\sum\limits_{i=4}^N |a_i|^2-8|a_1|^2|a_2|^2,$$
and the upper bound of $\N^2(|W\rangle_{A_1A_2A_3|A_4 \cdots A_N} $ (Eq.(\ref{thm4})) is
$$8|a_1|^2|a_2|^2+8|a_3|^2(|a_1|^2+|a_2|^2)+4(|a_1|^2+|a_2|^2+|a_3|^2)\sum\limits_{i=4}^N |a_i|^2.$$
When $a_1=a_2=0$, the  lower bound  and upper bound  of $\N^2(|W\rangle_{A_1A_2A_3|A_4 \cdots A_N})$  are equal.
\end{exm}

\noindent{\bf Discussion}\, \,
We have discussed the generalized monogamy relations of negativity for $N$-qubit systems. The generalized monogamy inequalities provide the lower and upper bounds of $\mathcal{N}(|\psi\rangle_{AB|C_1\cdots C_{N-2}})$ by using the CREN and the CRENOA. When the state $\mathcal{N}(|\psi\rangle_{AB|C_1\cdots C_{N-2}})$ has Schmidt rank two, Corollary 2
gives some  monogamy relations among $\mathcal{N}^2(|\psi\rangle_{AB|C_1\cdots C_{N-2}})$, $\mathcal{N}^2(|\psi\rangle_{A|BC_1\cdots C_{N-2}})$ and $\mathcal{N}^2(|\psi\rangle_{B|AC_1\cdots C_{N-2}})$.  Take, for example, the  $N$-qubit generalized $W$-class states (\ref{w-state}),  we calculate the lower and upper bounds of  $\mathcal{N}^2(|W\rangle_{AB|C_1\cdots C_{N-2}})$ and monogamy relations (\ref{cor2i})-(\ref{r3}).  We then generalize these results to $N$-qubit pure state under partition  $ABC_1|C_2\cdots C_{N-2}$.

Entanglement monogamy is a fundamental property of multipartite entangled states. The generalized monogamy relations  maybe test some higher-dimensional quantum systems. We believe that these generalized monogamy inequalities can be useful in quantum information theory. When we complete our paper, we find that the result (Theorem 2) in this paper is discussed in \cite{Tian}. But the proof in \cite{Tian} is valid only for Schmidt rank two. If Schmidt rank is not two, the theorem 2 in \cite{Tian} is not correct. In our paper, we have a coefficient beside the inequality. We use an example with Schmidt rank 3 to illustrate.
\begin{exm}
For a pure state $\ket \psi_{ABCD}$ in an four-qubit system:
\begin{equation}\label{}
  \ket \psi_{ABCD}=\frac{1}{\sqrt{3}}(\ket{0000}+\ket{0101}+\ket{1010}),
\end{equation}
the Schmidt rank of $ \ket \psi_{AB|CD}$ is $3$, and the  negativity
$ \mathcal{N}(\ket \psi_{AB|CD})=2.$

Besides, we have
$\rho_{AB}=\rho_{AD}=\rho_{BC}=\frac{1 }{3}(\ket{00}\bra{00}+\ket{01}\bra{01}+\ket{10}\bra{10})$,
and $\rho_{AC}=\rho_{BD}=\frac{1 }{3}(2\ket{00}\bra{00}+\ket{00}\bra{11}+\ket{11}\bra{00}+\ket{11}\bra{11})$.
Hence, we can obtain $\widetilde{\N}_a(\rho_{AB})=\widetilde{\N}_a(\rho_{AD})=\widetilde{\N}_a(\rho_{BC})=\frac{2}{3}$,
and
$\widetilde{\N}_a(\rho_{AC})=\widetilde{\N}_a(\rho_{BD})=\frac{\sqrt{3+2\sqrt{2}}}{3}+\frac{\sqrt{3-2\sqrt{2}}}{3}.$
A direct calculation shows that
$$\mathcal{N}^2(\ket \psi_{AB|CD})=4,$$
$$\frac{r(r-1)}{2}[2\widetilde{\N}_a^2(\rho_{AB})+\sum_{i=1}^{N-2}(\widetilde{\mathcal{N}}_a^2(\rho_{AC_i})
+\widetilde{\mathcal{N}}_{a}^2(\rho_{BC_i})]=3\times\frac{32}{9}=\frac{32}{3}\approx 10.67,$$
so inequality (\ref{Thm2}) hold. Without the  coefficient $\frac{r(r-1)}{2}$, the inequality
$\mathcal{N}^2(\ket \psi_{AB|CD})\leq [2\widetilde{\N}_a^2(\rho_{AB})+\sum_{i=1}^{N-2}(\widetilde{\mathcal{N}}_a^2(\rho_{AC_i})
+\widetilde{\mathcal{N}}_{a}^2(\rho_{BC_i})]$ in  \cite{Tian}, i.e.
$4\leq \frac{32}{9}$ does not hold.
\end{exm}

In paper  \cite{Tian}, they also discuss the lower bound of $\mathcal{N}^2(\ket \psi_{AB|C_1\cdots C_{N-2}})$ for a $N$-qubit pure state $\ket \psi_{AB|C_1\cdots C_{N-2}}$. Their  result was shown by theorem 3 in  \cite{Tian}:
\begin{equation}\label{tian lower}
\mathcal{N}^2(|\psi\rangle_{AB|C_1\cdots C_{N-2}}) \geq | \sum_{i=1}^{N-2}[\widetilde{\mathcal{N}}_a^2(\rho_{AC_i})-
\widetilde{\mathcal{N}}_{a}^2(\rho_{BC_i})]|.
\end{equation}
But an  important relation in their proof (\cite{Tian}, Eq.(47)) does not always hold.
Consider the following counter-example.
\begin{exm}
For the  four-qubit pure state $\ket \psi_{ABCD}=\frac{1}{\sqrt{2}}(\ket{0000}+\ket{1011})$,
we have $\mathcal{N}^2(|\psi\rangle_{AB|CD})=(2\times \frac{1}{\sqrt{2}}\times \frac{1}{\sqrt{2}})^2=1 $.
In addition, $\rho_{AC}=\rho_{AD}=\frac{1}{2}(\ket{00}\bra{00}+\ket{11}\bra{11})$,
$\rho_{BC}=\rho_{BD}=\frac{1}{2}(\ket{00}\bra{00}+\ket{01}\bra{01})$. A direct calculation give that
$$\widetilde{\N}_a(\rho_{AC})=\widetilde{\N}_a(\rho_{AD})=1, \hspace{5mm}  \widetilde{\N}(\rho_{AC})=\widetilde{\N}(\rho_{AD})=0,$$
$$\widetilde{\N}_a(\rho_{BC})=\widetilde{\N}_a(\rho_{BD})=0, \hspace{5mm}  \widetilde{\N}(\rho_{BC})=\widetilde{\N}(\rho_{BD})=0.$$
Putting these values into Eq.(\ref{tian lower}), we get $1\geq 2$, this is a contradiction.

Our lower bound of $\mathcal{N}^2(|\psi\rangle_{AB|CD})$ is
$$\max \{\sum_{i=1}^{2}[\widetilde{\mathcal{N}}^2(\rho_{AC_i})-\widetilde{\mathcal{N}}_{a}^2(\rho_{BC_i})], \sum_{i=1}^{2}[\widetilde{\mathcal{N}}^2(\rho_{BC_i})-\widetilde{\mathcal{N}}^2_{a}(\rho_{AC_i})]\}
=0.$$

\end{exm}

\vspace{2.5ex}
\noindent{\bf Acknowledgments}\, \,
This work is supported by the NSFC under numbers 11405060, 11475178 and 11571119. It is a pleasure to thank Prof. Shao-Ming Fei for helpful discussion and opinion.

\end{document}